%% file: main.tex
\documentclass[twoside]{article}
\usepackage[letterpaper,margin=1in]{geometry}
\usepackage{amsmath,graphicx}
\usepackage{mathtools}
\usepackage{caption}
\usepackage{subcaption}
\usepackage{tikz}
\usepackage[colorinlistoftodos,prependcaption,textsize=tiny]{todonotes}
\usepackage{wrapfig}
\usepackage{hyperref}
\usepackage{cleveref}
\usepackage{thm-restate}
\usepackage{comment}

\tikzset{
    cross/.pic = {
    \draw[rotate = 45] (-#1,0) -- (#1,0);
    \draw[rotate = 45] (0,-#1) -- (0, #1);
    }
}
\usepackage[colorinlistoftodos,prependcaption,textsize=tiny]{todonotes}

\usepackage[boxruled,vlined]{algorithm2e}
\usetikzlibrary{positioning}
\usepackage{comment}

\usetikzlibrary{positioning}
\usepackage[charter]{mathdesign}
\usetikzlibrary{shapes.symbols}
\usetikzlibrary{calc}

\newtheorem{theorem}{Theorem}[section]
\newtheorem{lemma}[theorem]{Lemma}

\newtheorem{definition}[theorem]{Definition}

\newtheorem{prop}[theorem]{Property}

\newenvironment{proof}{{\bf Proof:}}{\hfill\rule{2mm}{2mm}}

\newcommand{\TT}{\mathcal{T}}
\newcommand{\OO}{O}
\newcommand{\GG}{G}
\newcommand{\GAT}{{\widehat{G}}}
\newcommand{\GT}{{\widehat{G}}}
\newcommand{\LCA}{\textsc{Lca}}
\newcommand{\MM}{M}
\newcommand{\NN}{N}

\newcommand{\DIA}{\diamond}
\newcommand{\VV}{V}
\newcommand{\EE}{E}

\newcommand{\VM}{\VV_{\MM}}

\newcommand{\GM}{\GG_{\MM}}
\newcommand{\GN}{\GG_{\NN}}

\newcommand{\PP}{\mathcal{P}}

\newcommand{\DP}{\textsc{Dp}}
\newcommand{\SSR}{\textsc{Ssrp}}
\newcommand{\TL}{\widetilde{O}}
\newcommand{\RR}{\mathcal{R}}
\newcommand{\FF}{F}
\newcommand{\TTT}{\mathcal{T}}
\newcommand{\HH}{H}
\newcommand{\QU}{\textsc{Query}}
\newcommand{\QD}{\textsc{Query-DEP}}
\newcommand{\DO}{\textsc{Do}}
\newcommand{\SDO}{\textsc{Sdo}}
\newcommand{\DEP}{\textsc{Dep}}
\newcommand{\BMM}{\textsc{Bmm}}
\newcommand{\UNIQUE}{\textsc{Unique}}
\newcommand{\Det}{\textsc{Detour}}
\newcommand{\SPT}{\textsc{Spt}}
\newcommand{\SM}{\setminus}
\newcommand{\BFS}{\SPT}

\title{Near Optimal Algorithm for Fault Tolerant Distance Oracle and Single Source Replacement Path problem 
}

\author{
Dipan Dey \\
  IIT Gandhinagar \\
  Gandhinagar \\
  India\\
  \texttt{dey\_dipan@iitgn.ac.in} \\
  \and
  Manoj Gupta \\
  IIT Gandhinagar \\
  Gandhinagar \\
  India\\
  \texttt{gmanoj@iitgn.ac.in} \\
}
\begin{document}

\maketitle

\input{abstract}

\input{intro}
\input{prelim}
\input{overview}

\input{algorithm}

\input{departing}


\bibliographystyle{plainurl}
\bibliography{paper}
\input{appendix}

\end{document}

%% file: abstract.tex
\begin{abstract}
In a graph $\GG$ with a source $s$, we design a distance oracle that can answer the following query: $\QU(s,t,e)$ -- find the length of shortest path from a fixed source $s$ to any destination vertex $t$ while avoiding any edge $e$. We design a deterministic algorithm that builds such an oracle in $\TL(m\sqrt n)$ time\footnote{$\TL()$ hides poly$\log n$ factor }. Our oracle uses $\TL(n\sqrt n)$ space and can answer queries in $\TL(1)$ time. Our oracle is an improvement of the work of Bil\`{o} et al. (ESA 2021) in the preprocessing time, which constructs the first deterministic oracle for this problem in $\TL(m\sqrt n+n^2)$ time.

Using our distance oracle, we also solve the {\em single source replacement path problem} ($\SSR$ problem).
Chechik and Cohen (SODA 2019) designed a randomized combinatorial algorithm to solve the $\SSR$ problem. The running time of their algorithm is $\TL(m\sqrt n + n^2)$. In this paper, we show that the $\SSR$ problem can be solved in $\TL(m\sqrt n + |\RR|)$ time, where $\RR$ is the output set of the $\SSR$ problem in $\GG$. 
Our  $\SSR$ algorithm is optimal (upto polylogarithmic factor) as there is a conditional lower bound of $\Omega(m\sqrt n)$ for any combinatorial algorithm that solves this problem. 
\end{abstract}

%% file: intro.tex
\section{Introduction}
Real-life graph networks are prone to failures, e.g., nodes or links can fail. Thus, algorithms developed for these networks must be resilient to failures. For example, there may be some edges or links which are not working in the network and we want to avoid them. In this paper, we present an algorithm to create an oracle for the single source shortest path problem in a fault-prone graph. Such algorithms are also called fault-tolerant algorithms.

Consider an undirected and unweighted graph $\GG$ with a source $s$. We want to build an oracle that can find the length of shortest path from $s$ to any other vertex in the presence of faulty edges -- such an oracle is also called a {\em fault-tolerant distance oracle}. Formally,

\begin{definition}
        A fault-tolerant distance oracle answers the following query in a graph $\GG$:

        \begin{center}
                $\QU(s,t,\FF)$: Find the length of  shortest path from  $s$ to $t$ avoiding the set $\FF$ of edges.
        \end{center}

\end{definition}

The time it takes to answer a query is called the {\em query time}. If the query is always from a fixed source $s$ and $|\FF|\le f$, then the distance oracle is called a $f$-edge fault tolerant single source distance oracle, or $\SDO(f)$ in short. If all vertices can be sources, the oracle is called $f$-edge fault tolerant distance oracle, or $\DO(f)$. We list some results related to distance oracles:

Demetrescu et al. \cite{Demetrescu2008} designed a $\DO(1)$ with  $\TL(n^2)$ space and $O(1)$ query time. Bernstein and Karger\cite{Bernstein2009} showed that this oracle can be built in $\TL(mn)$ time.   Pettie and Duan \cite{DuanP10} extended the result of Demestrescu et al. to two faults. They designed a $\DO(2)$ with $\TL(n^2)$ space and $\TL(1)$ query time. Gupta and Singh \cite{GuptaS18} designed a $\SDO(1)$ with $\TL(n\sqrt n)$ space and $\TL(1)$ query time.  Recently  Bil{\`{o}} et al. \cite{BiloCFS21} built the $\SDO(1)$ (described in \cite{GuptaS18}) in  $\TL(m\sqrt n + n^2)$ time.
Many different aspects of distance oracles  have been studied in literature  \cite{Bernstein2008,Bernstein2009,Bilo2016,ChechikCFK17,Chechik2010,10.5555/1496770.1496826,10.1145/2438645.2438646}.

In this paper, we will focus our attention on building $\SDO(1)$. Due to Bil{\`{o}} et al. \cite{BiloCFS21}, the time to build $\SDO(1)$ is $\TL( m\sqrt n + n^2)$. Chechik and Cohen \cite{ChechikC19} showed that, the first term in this running time is a conditional lower bound for $\SSR$ problem. But it is not clear if the second term is necessary. 
In this paper, we build a $\SDO(1)$ in $\TL(m\sqrt n)$ time -- this preprocessing algorithm has a better runtime than \cite{BiloCFS21} for sparse graphs, which is state of the art for this problem till now. Using our $\SDO(1)$ data structure, we are able to reduce the runtime of the algorithm solving $\SSR$ problem too. Our distance oracle is quite different from the distance oracle of Gupta and Singh \cite{GuptaS18} -- though we use the main technical idea of \cite{GuptaS18} crucially in our paper too. The construction of this new oracle is the main technical result of this paper. 

\begin{theorem}
        \label{thm:sdo}
        For undirected, unweighted graphs there is a deterministic algorithm that can build a $\SDO(1)$ of size $\TL(n\sqrt n)$ and query time $\TL(1)$ in $\TL(m\sqrt n)$ time.
\end{theorem}

\subsection{Application: Single Source Replacement Path Problem}
Let us first look at the  {\em replacement path problem}. In this problem, we are given a source $s$ and a destination $t$. We assume that there is a unique shortest path from $s$ to $t$, denoted by $st$. 
\begin{definition}
	(Replacement Path Problem) Let $s$ be a source and vertex  $t$ be the destination in $\GG$.
	For each $e \in st$ path, output the length of the shortest path from $s$ to $t$ avoiding $e$. 
\end{definition} 

The replacement path problem was first investigated due to its relation with auction theory \cite{Hershberger2001, NisanR01} and has been studied extensively. For an undirected graph with  non-negative edge weights, the replacement path problem can be solved in  $\TL(m+n)$ time\cite{MalikMG89,Hershberger2001,NardelliPW03}.
We look at the generalization of the replacement path problem -- the single source replacement path problem.
\begin{definition} ($\SSR$ problem)
	Let $s$ be a source in a graph $\GG$ which is undirected and unweighted. For each vertex $t \in \GG$ and  each $e \in st$ path, output the length of the shortest path from $s$ to $t$ avoiding $e$.
\end{definition}
Chechik and Cohen \cite{ChechikC19} designed a randomized combinatorial algorithm that solves the $\SSR$ problem in $\TL(m\sqrt n +n^2)$ time. They also showed a matching conditional lower bound via Boolean Matrix Multiplication. 
\begin{lemma} \cite{ChechikC19}
        \label{lem:lowerbound}
        Let $\BMM(n,n)$ be the time taken to multiply two $n \times n$ boolean matrices with a total of $m$ ones. Under the assumption that any combinatorial algorithm for $\BMM(n,n)$  requires $mn^{1-o(1)}$ time\footnote{In a RAM\ model with words of $O(\log n)$ bits.}, any combinatorial algorithm for $\SSR$ problem requires $\Omega(m\sqrt n)$ time.
\end{lemma}

It may seem that the algorithm of Chechik and Cohen \cite{ChechikC19} is nearly optimal. It is indeed the case if the output size is $O(n^2)$. However, for a low-diameter graph, this extra additive factor seems unnecessary. If the graph is dense ($m \ge n^{3/2}$), then the $n^2$ factor is subsumed by the first term $m\sqrt n$. Thus, when $m < n^{3/2}$ and the graph has a low diameter, can we improve the running time of the $\SSR$ problem? For such a graph, the algorithm of Chechik and Cohen \cite{ChechikC19} is not optimal.
Similar to \cite{ChechikC19}, Gupta et al. \cite{GuptaJM20} also designed an algorithm for the $\SSR$ problem. 
Even this algorithm has the running time $\TL(m\sqrt n+n^2)$ -- though it uses an entirely different approach compared to \cite{ChechikC19}.
Thus, the main question is: {\em can we remove this extra additive factor of $n^2$ from the running time of the $\SSR$ problem?} In this paper, we design such an algorithm:

\begin{theorem}
        \label{thm:main}
        There is a deterministic algorithm for $\SSR$ problem with a running time of $\TL(m\sqrt n+ |\RR|)$ where $|\RR|$ is the output size of $\SSR$ problem in $\GG$.
\end{theorem}

In the above theorem, $|\RR|$ is the output size, thus an implicit lower bound on the $\SSR$ problem. Using  \Cref{lem:lowerbound}, we conclude that our algorithm is nearly optimal up to a polylogarithmic factor. 

To build an algorithm for the $\SSR$ problem, we first build a $\SDO(1)$. Then, for each $t \in \GG$ and each $e \in st$ path, we  call $\QU(s,t,e)$ and output the answer. Thus, we claim the following lemma:

\begin{lemma}
        \label{lem:reduction}
        If we can build a $\SDO(1)$ with query time $q$ in time $T$, then there is an algorithm for $\SSR$ problem with a running time $O(T+q|\RR|)$, where $|\RR|$ is the output size in $\GG$.
\end{lemma}

The above lemma, along with \Cref{thm:sdo} implies \Cref{thm:main}. 

\subsection{Related Work}

Other related problems include the fault-tolerant subgraph problem.
In this problem, we want to find a subgraph of $\GG$ such that the shortest path from
$s$ is preserved in the subgraph after any edge deletion.
Parter and Peleg \cite{Parter2013} designed an algorithm to compute a single fault-tolerant subgraph with
$ O(n^{3/2})$ edges. They also showed that their result could be easily extended to multiple sources. This result was later extended
to dual fault by Parter [16] with $ O(n^{5/3}) $ edges. Gupta and Khan \cite{GuptaK17} extended
the above result to multiple sources.
All the above results are optimal due to a result by Parter \cite{Parter2015} which states that a multiple source $f$-fault tolerant  subgraph requires $\Omega\Big(n^{2-\frac{1}{f+1}}\Big)$
edges.  Bodwin et al. \cite{BodwinGPW17} showed the existence of a $f$-fault tolerant
subgraph of size $O\Big( fn^{2-\frac{1}{2^f}} \Big)$.

%% file: prelim.tex
\section{Preliminaries}
Let $\GG(\VV,\EE)$ be an undirected unweighted graph with a source $s$. 
Given two vertices $u$ and $v$ in a graph $H$, unless otherwise stated,  $(uv)_H$ denotes the shortest path from $u$ to $v$ in $H$. If $H = \GG$, we will remove the subscript and the brackets -- we will apply this policy for all the notations below.  $|uv|_H$ denotes the length of the shortest path in $H$. Some of our graphs will be weighted, even though $\GG$ is unweighted. If $H$ is weighted, then we will abuse notation and use  $|uv|_H$ to denote the weight of the shortest path from $u$ to $v$ in $H$.
 The edges and vertices of $H$ will be denoted by $\EE_H$ and $\VV_H$, respectively.\ Additionally, $m_H$ and $n_H$ will denote the number of edges and vertices in $H$, respectively. 
 $\SPT_H(s)$ denotes the shortest path tree from $s$ in $H$. We can view the $\SPT_H(s)$ to be drawn from top to bottom with the top vertex being $s$. For any two vertices $u,v$ on the $st$ path (the path between $s$ and $t$) in $\SPT_H(s)$, we say that   $u$ is before / above $v$ if  $|su|_H <|sv|_H$. Similarly, we say that  $u$ is after/below $v$ if $|su|_H >|sv|_H$.
   For an edge $e$ in a  weighted graph $H$, $wt_H(e)$ will denote the weight of $e$.
Given two paths $(uv)_H$ and $(vw)_H$, the path $(uv)_H +(vw)_H$ denotes their concatenation.
$P[u,v]$ denotes a contiguous subpath of $P$  starting at $u$ and ending at $v$.
Sometimes, we may also write {\em the interval $[u,v]$ of $P$} to denote   $P[u,v]$. We say {\em $u$ comes before $v$} on a path $R$ starting from $s$, if $|R[s,u]| < |R[s,v]|$. Similarly, we can define the term {\em $u$ comes after $v$} on path $R$.  

A replacement path $R$ is the shortest path from $s$ to $t$ avoiding an edge $e$ on $st$ path. There can be many replacement paths of the same length avoiding $e$. 
To ensure uniqueness, we  will use the following definition of  replacement path\footnote{This was referred to as {\em preferred replacement path} in \cite{GuptaK17}.}.

\input{fig_detour}
\begin{definition} \label{def:replacementpath} ( Replacement Path) A path $R$ from $s$ to $t$ avoiding $e$ is called a  replacement path if (1) it diverges from and merges to the $st$ path just once (2) its divergence point from the $st$ path is as close to $s$ as possible. (3) it is the {\bf lexicographically smallest}\footnote{Let $P$ and $P'$ first diverge from each other to $x \in P$ and $x' \in P'$ respectively. If the index of $x$ is lower than  $x'$, then $P$ is said to be lexicographically smaller than $P'$.}  shortest path in $\GG$ satisfying (1) and (2). 

\end{definition}

We now define some terms related to  replacement paths.
$(st \DIA  e)_H$ denotes the  replacement path from $s$ to $t$ avoiding the edge $e$ in $H$.  We can generalize this notation to a  replacement path that avoids a set of edges. Thus, $(st \DIA \FF)_H$ denotes the replacement path from $s$ to $t$ in $H$ avoiding a set $\FF$ of edges. 
In our algorithm, after we find the replacement path $(st \DIA e)_H$, we will store its length in $d_H(s,t,e)$.   Sometimes, we also want to store the length of $(st \DIA \FF)_H$. In that case, we will store it in $d_H(s,t,\FF)$. 
     
     \begin{definition}
     \label{def:detour}
     	(Detour and Detour point of a  replacement path) \cite{ChechikM20}
  Let $R = st \DIA e$. Then, the detour of $R$ is $R \setminus st$. That is, let us assume that $R$  leaves $st$ above $e$ at a vertex $u$, and merges back on $st$ at vertex $v$ after $e$, then detour of $R$ is $R[u,v]$. Also, the vertex at which the detour starts is called the detour point of $R$. So, $u$ is the detour point of $R$ or in short $\DP(R)=u$.
     \end{definition}

Lastly, in our algorithm, we will need to find least common ancestor of any  two vertices $u$ and $v$ in $\SPT_H(s)$. Let $\LCA_H(u,v)$ denotes the least common ancestor of $u$ and $v$ in $\SPT_H(s)$.
    To find the $\LCA$, we will use the following result:

    \begin{lemma} (See \cite{BenderF00} and its references)
\label{lem:lca}
    Given a tree $T$ on $n$ vertices, we can build a data structure of size $O(n)$ in $O(n)$ time such that the least common ancestor query can be answered in $O(1)$ time.
    \end{lemma}

%% file: fig_detour.tex
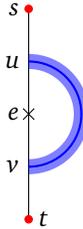
\begin{figure}[hpt!]

\centering
\begin{tikzpicture}[scale=0.7]
\coordinate (s) at (0,8);
\coordinate (t) at (0,4);
\coordinate (u) at (0,7);
\coordinate (v) at (0,5);
\coordinate (e) at (0,6);

\draw (s) node[left] () {$s$} -- (t) node[right] () {$t$};
\filldraw [red] (s) circle (2pt);
\filldraw [red] (t) circle (2pt);
\draw  (u) arc(90:-90:1) [ thick,blue];
\draw  (u) arc(90:-90:1) [ line width=2mm,blue,opacity=0.5];
\draw (e) node[left] () {$e$} pic[rotate = 0] {cross=3pt};
\draw (u) node [left] {$u$};
\draw (v) node[left] {$v$};
\end{tikzpicture}
\caption{Replacement path with detour $uv$ and detour point $u$.}
\end{figure}

%% file: overview.tex
\section{Overview of our algorithm to build SDO(1)}


We will use divide and conquer approach to build $\SDO(1)$. This strategy has been previously used for directed graphs in \cite{GrandoniW12,ChechikM20}. However, simply using this strategy will not get us  close to our desired bound of
$\TL(m\sqrt n)$. For that, we need to combine this divide and conquer strategy with an idea of Gupta and Singh \cite{GuptaS18}. This combination is one of the technical contributions of the paper.

Like \cite{GrandoniW12,ChechikM20}, we use the following   separator lemma to divide the graph $\GG$.
\begin{theorem}(Separator Lemma \cite{GrandoniW12,ChechikM20})
\label{lem:separator}
Given a tree $T$  with  $n$ nodes rooted at a source  $s$, one can find in $O(n)$ time a vertex $r$ that separates the tree  $T$ into  edge disjoint sub-trees $\MM,\NN$ such that  $\EE_{\MM}\cup \EE_{\NN}=\EE_{T}, \VV_{M} \cap \VV_{N}=\{r\}$  and $\frac{n}{3 } \leq | \VV_{\MM}|, |\VV_{\NN}| \leq \frac{2n}{3}$.
  \end{theorem}

 Without loss of generality, we will assume that $s\in \VM$. Thus,  $r$ is the root of $\NN$. Also, note that $s$ and $r$ may or may not be the same.  Let $\GG_{\MM}$ and $\GG_{\NN}$ be the graph induced by the vertices of $\MM$ and $\NN$, respectively. 
 There is one more important term that we will use in our paper:

\begin{definition} (Primary Path $\PP$)
Using the separator lemma,  $\BFS(s)$ can be divided into two sub-trees $\MM$ and $\NN$ with roots $s$ and $r$. The path from $s$ to $r$ is called the primary path and is denoted by $\PP$.
\end{definition}


We now describe our data structure that we will build recursively. We can view the data structure as a binary tree $\TTT$. The root contains data structure for the entire graph $\GG$. We will abuse notation and say the root is the graph $\GG$.

The left child of the root will contain the graph $\GM$ and some weighted edges -- we will describe the utility of these weighted edges in the next section. These weighted edges are not in $\GG$ but are added by our algorithm.  
We will then build a data structure for $\GM$  recursively.
The right child of the root will contain the graph $\GN$, again with some weighted edges. 

At the root of $\TT$, we store the following data structures. For each $v \in \GG$, set $d(s,v)=|sv|$. Similarly, set $d(r,v)= |rv|$.  For each $v \in \GN$, set $d(s,v,\GN) = |sv \DIA \GN|$. Similarly, for each $v \in \GM$, set $d(s,v, \GM) = |sv \DIA \GM|$. All these quantities can be computed using a single source shortest path algorithm in $\TL(m+n)$ time. Additionally, we will find the length of all replacement paths from $s$ to $r$ avoiding edges on the primary path $\PP$. This can be done in $\TL(m+n)$ time using\footnote{This algorithm work for graphs with non-negative edge weights. And our graph may have weighted edges.} \cite{Hershberger2001}. We will set $d(s,r,e) = |sr \DIA e| $ for each edge $e \in \PP$.

We store the above data-structure in each node of $\TT$. If a node of $\TT$ contains graph $H$, then we  can contruct the above data-structures in $\TL(m_H + n_H)$ time.  We now describe our algorithm that finds replacements paths using $\TT$.

%% file: algorithm.tex
\input{fig_departing}

Let us see how we find and store lengths of the replacement paths at the root of $\TT$, that contains graph $\GG$. First, we find the replacement paths for edges on the primary path. Let $R = st \DIA e$ where $e \in \PP$.  We define $R$ to be either {\em jumping} or {\em departing} depending on whether it merges back to the primary path or not.

\begin{definition}(Jumping and Departing paths)
   Let $R = st \DIA e$ where $e \in \PP$ . $R$ is called a \textbf{\em jumping path} if it uses some vertex $u \in \PP$  after $e$. If the path is not jumping then it is a \textbf{\em departing path}. If a replacement path is jumping, then it is called {\em jumping replacement path}. Similarly, we define departing replacement path. See \Cref{fig:departingjumping}
for a visualization of these two kinds of paths.\end{definition}

Note that, jumping or departing path is defined only when the edge fault is on the primary path. Also,  if a replacement path is departing, then the destination $t$ cannot lie on $\PP$. In \Cref{sec:detail}, we will find all jumping replacement paths.

In  \Cref{sec:departing}, we design a new algorithm for finding and storing all departing replacement paths. To this end, we will use the main idea in the paper of Gupta and Singh \cite{GuptaS18}.  In \cite{GuptaS18}, the authors  sampled a set of  vertices  with probability of $\OO(\frac{1}{\sqrt{n}})$. Then, for a vertex $t \in \GG$, they find a sampled vertex near  $t$ on the $st$ path. They call this vertex  $t_s$. Then, they  show the following important lemma, which is the main idea of their paper:

\begin{lemma} (Lemma 11 in \cite{GuptaS18})
\label{lem:guptaaditi}
The number of replacement paths  from $s$ to $t$ that avoid edges in $st_s$ path and also avoid $t_s$  is $O(\sqrt n)$.
\end{lemma}

An astute reader can see that the definition of replacement paths in the above definition looks very similar to departing replacement paths. We prove that this is indeed the case. Thus, we can transfer the result in \Cref{lem:guptaaditi} to departing replacement paths. This is the main novelty of the paper. The main technical result of \Cref{sec:departing} is as follows:

\begin{lemma}
\label{lem:departing}
For each  $t \in \GG$, all departing replacement paths to $t$ can be found in deterministic $\TL(m\sqrt n)$ time. Moreover, the length of all such departing paths can be stored in a data structure of size $\TL(n\sqrt n)$ and can be queried in $\TL(1)$ time.
\end{lemma}

\input{fig_gmgn}
\section{Algorithm to build SDO(1): Replacement paths that are not departing}
\label{sec:detail}
In \Cref{sec:departing}, we will find all and store all departing replacement paths.
Thus, we just need to concentrate on the replacement paths that are either jumping or the faulty edge $e \notin \PP$.
We now divide remaining replacement paths depending on where the destination $t$ and faulty edge $e$ lies. There are  following cases:

\subsection{$e \in \GM$ and $t \in \GN$}
\label{sec:gmgn}

This case itself can be divided into two cases depending on whether $e$ lies on the primary path or not.
\begin{enumerate}
    \item $e \in \PP$ (See  \Cref{fig:gmgn}(a))

Let $R = st \DIA e$.
If $R$ is departing then we will see how to find it in  \Cref{sec:departing}. So, assume that $R$ is jumping . This implies that $R$ merges back to $\PP$ at a vertex, say $w$, after the edge $e$. Since $t \in \GN$, $st = sw + wr + rt$.  Thus, after merging with $\PP$ at $w$, the replacement path passes through $r$. In that case, $|st \DIA e| = |sr \DIA e| + |rt|$. We can easily find the right hand side of the above equality as we have stored $d(s,r,e) = |sr \DIA e|$ and $d(r,t) = |rt|$.

\item $e \notin \PP$ (See  \Cref{fig:gmgn}(b))

In this case, we claim that $st \DIA e = st$. The $st$ path has $\PP$ as its prefix. Since $\PP$ lies in $\GM$ and survives after the deletion of $e$, $st$ path remains intact.

\end{enumerate}

\subsection{$e \in \GM$ and $t \in \GM$}
\label{subsec:MM}
\input{fig_gmgmcase1}

Since both $e$ and $t$ lie in $\GM$, one may think that we can recurse our algorithm in $\GM$ to find $st \DIA e$.
If $st \DIA e$ completely lies inside $\GM$, this is indeed the case. However,  $st \DIA e$ may also use edges of $\GN$.
To handle such cases, before recursing in $\GM$, we will add weighted edges to it. For each $v \in \GM$, we will add an edge from $r$ to $v$ with a weight $|rv \DIA \GM|$.
We have already calculated this weight, it is stored in $d(r,v,\GM)$.  Let the set of weighted edges added to $\GM$  be called $X$. 
We now look at two cases, (1) $e \in \PP$ and (2) $e \notin \PP$.

\subsubsection{$e \in \PP$ }
\label{sec:gmgmp}
Let $R = st \DIA e$ be a jumping replacement path.
 We will show that $st \DIA e = sr \DIA e + rt$. As we have calculated the length of both the paths in the right-hand side of the above equality, there is no need to even recurse in this case. To prove the above equality, we first prove the following simple lemma:

\begin{lemma}
\label{lem:passr}
Let $e \in  \PP$,  $t \in \GM$. Assume that   the jumping replacement path $R =st \DIA e$ uses some edges of $\GN$. Then $st \DIA e$ passes through $r$.
\end{lemma}

\begin{proof}
Since $R$ is jumping, it merges with $\PP$.  There are two ways in which $R$ can merge with $\PP$.
\begin{enumerate}
    \item $R$ merges with $\PP$ and then visits the edges of $\GN$.

    Let us assume that $u$ is the last vertex of $\GN$ in the path  $R$ and $R$  merges with $\PP$ at $w$. Since $R$ first merges with  $\PP$ and then  visits the  edges of $\GN$, $u$ comes after $w$ on $R$. Since $w$ is below $e$ on $\PP$, we claim that the sub-path $wu$ of $su$ survives in $\GG\setminus e$ and is also the shortest path from $w$ to $u$. But  $wu$ path passes through $r$. Thus, $st \DIA e$ passes through $r$ by construction in \Cref{lem:separator}.






    \item $R$ visits an edge of $\GN$ and then merges with $\PP$ (See  \Cref{fig:gmgm}(a))


    In this case, we will show that $R$ merges with $\PP$ at $r$. For contradiction, let $w$ be the  vertex at which $R$ merges with $\PP$  such that $w  \neq r$.
 Let $u$ be the first vertex of $\GN$ visited by $R$.
  Also, $w$ lies after $u$ on path $R$. Thus, the replacement path $R = R[s,u]+ R[u,w] + R[w,t]$.  Since $w$ lies below $e$ on $\PP$, $wu$ sub-path of $su$ does not contain $e$ and is also the shortest path from $w$ to $u$. Thus, $R[u,w] = uw$. But $uw$ path passes through $r$. This implies that $R$ merges with $\PP$ at $r$ contradicting our assumption that $w \neq r$.
\end{enumerate}

\end{proof}
 We are now ready to prove the main lemma in this subsection.

\begin{lemma}
\label{lem:ste}
Let $e \in \PP$ and   $t \in \GM$. Assume that   the jumping replacement path $R =st \DIA e$ uses some edges of $\GN$. Then $|st \DIA e| = |sr \DIA e| + |rt|$.
\end{lemma}

\begin{proof}
Using  \Cref{lem:passr}, $R$ passes through $r$. So, we have $R = R[s,r]+ R[r,t]$. The first  summand on the right hand side of the above equality represent a path from $s$ to $r$ avoiding $e$. Thus, $|R[s,r]| = |sr \DIA e|$.

We will now show that $|R[r,t]| = |rt|$. Clearly $|R[r,t]| \ge |rt|$ as the first path avoid $e$ and the second path may or may not. We will now show that the second path also avoids $e$ which will imply that both paths are of same length.  For contradiction, assume that $rt$ passes through $e$. Let us assume that there is a vertex $w$ before the edge $e$ on path $\PP$  such that $rt = rw + wt$. Thus,  $rw$ passes through $e$  but $wt$ avoids $e$. But then, there is a path $R'$ such that $R' = sw + wt$ which avoids $e$. We claim that $|R'| < |R|$ contradicting the fact that $R$ is the replacement path from $s$ to $t$ avoiding $e$.

To this end,
\begin{tabbing}
\hspace{30mm}$|R|$\= $= |sr \DIA e| + |R[r,t]|$\\
\>$\ge |sr \DIA e| + |rt|$\\
\>$=  |sr \DIA e|  + |rw| + |wt|$\\
\> $\ge |sr \DIA e|+|wt|$\\
Since $|sw| < |sr|$, $|sw| < |sr \DIA e|$\\
\>$> |sw| + |wt|$\\

\>$= |R'|$.
\end{tabbing}
\
This completes the proof of the lemma.
\end{proof}


\subsubsection{$e \notin \PP$ }
\label{sec:gmgnnp}
In this case, we will show a path in $\GM \cup X$ such that it avoids $e$ and has the same length as $st \DIA e$. Please see  \Cref{fig:gmgm}(b) for a visualization of this case.

\begin{lemma}
\label{lem:enotinP}
Let $e \in \GM \SM \PP$ and   $t \in \GM$. Assume that   the replacement path $R =st \DIA e$ uses some edges of $\GN$. Then there is a path in $\GM \cup X$ that avoids $e$ and has length $|st \DIA e|$.
\end{lemma}
\begin{proof}
Let us first prove that  $R$ can alternate between edges of $\GN$ and $\GM$ just once. To this end,  let   $u $ be the last vertex of $\GN$ visited by $R$. Thus, $R = R[s,u] + R[u,t]$. But the shortest path $su$ remains intact in $\GG \setminus e$ as $e \notin \PP$ and the shortest $su$ path passes through $r$. This implies that $R = R[s,u] + R[u,t] = R[s,r] +\ R[r,u] +\ R[u,t]$. By construction, the first path $R[s,r] = sr$ and it completely lies in $\GM$. The second path $R[r,u]=ru$ and it completely lies in $\GN$. Let $v$ be the vertex just after $u$ in $R$. So, $v \in \GM$. So we have $R =sr +\ ru +\ (u,v)+R[v,t]$. By construction,  $R[v,t]$ completely lies in $\GM$. Thus, $R$ alternates from edges of $\GN$ to $\GM$\ just once.

From the above discussion $R = R[s,r] + R[r,u] + (u,v) + R[v,t] = R[s,r] + R[r,v]+ R[v,t]$.   The first and the last paths of the above equality completely lies in $\GM$. By construction, $R[r,v]$ does not contain any edge of $\GM$. Thus, $R[r,v]$ is the shortest path from $r$ to $v$ avoiding edges of $\GM$, that is $|R[r,v]| =|rv \DIA \GM| = d(r,v,\GM)$.

Since  we have added an edge from $r$ to $v$ with weight $d(r,v,\GM)$, we will now show that there is a path in $\GM \cup X$ that avoids $e$ and has same weight as $R$. Consider the path $R' = R[s,r] + (r,v) + R[v,t]$. The reader can check all the three subpaths lie completely in $\GM \cup X$. Moreover, $(r,v) \in X$ is a weighted edge. Thus the weight of $|R'|  =|R[s,r]| + wt_{(\GM\cup X)}(r,v)+|R[v,t]| =|R[s,r]| + d(r,v,\GM)+|R[v,t]|=|R[s,r]| + |R[r,v]|+|R[v,t]|=|R|$. 
 This completes the proof.


\end{proof}

\subsection{$e \in \GN$ and $t \in \GM$}
\label{sec:gngm}
In this case, $st$ path completely lies in $\GM$ and thus survives. Thus, $|st \DIA e| = |st| = d(s,t)$.

\subsection{$e \in \GN$ and $t \in \GN$}
\label{subsec:NN}

In this case, we will recurse in $\GN$. However, $\GN$ may not contain the source $s$ if $s \neq r$. In that case, before recursing, we add a new source $s_{\NN}$ in $ \GN$.
We also  add some {\em weighted }edges to $\GN$. For each $v \in \GN$, we add an edge from $s_{\NN}$ to $v$ with a weight $d(s,v,\GN) = |sv \DIA \GN| $.  Let this set of edges be called $Y$.
Let us now show that we will find all the replacement paths if we recurse in $\GN \cup Y$.

\begin{lemma}
\label{lem:4.4}
Let $e \in \GN$ and $t \in \GN$. There is a path from $s_{\NN}$ to $t$ in $\GN \cup Y$   that avoids $e$ and has weight $|st \DIA e| $.
\end{lemma}
\begin{proof}
\label{GNPath}

Let $R =st \DIA e$.
Let us first prove that once  $R$ visits an edge of  $\GN$, it cannot visit an edge of $\GM$ anymore. Let $u$ be the last vertex of $\GM$ visited by $R$ and $v$ be the vertex after $u$ in $R$. Thus, $v \in \GN$.  Thus, $R = R[s,u] +(u,v)+ R[v,t]$. By construction, $R[v,t]$ completely lies in $\GN$. The shortest path $su$ remains intact in $\GG \setminus e$ as $e \in \GN$ and the path $su$ completely lies in $\GM$. Thus, $R[s,u] = su$ . So, the path $R$ first visits only edges of $\GM$(in $R[s,u])$, then goes to $\GN$ (by taking the edge $(u,v)$) and then remains in $\GN$ (in $R[v,t]$). Thus, $R$ does not visit any edge of  $\GM$ once it visits an edge of $\GN$.

By the above discussion,  $R = R[s,u] +\ (u,v)  + R[v,t] = R[s,v] +R[v,t]$ such that   $R[v,t]$  completely lies  in $\GN$. Also,  $R[s,v]$ completely lies in $\GM$ except the last edge which has one endpoint in $\GM$ and other endpoint $v \in \GN$. 
Thus, $R[s,v] = |sv \DIA \GN|= d(s,v,\GN)$ and we have $R =  sv \DIA \GN +\ R[v,t]$.

Now consider a path $R'$ that avoids $e$ from $s_{\NN}$ to $t$ in  $\GN \cup Y$. We construct this path as follows: we will first take the weighted edge $s_{\NN} \to v$ and then the path $R[v,t]$. The reader can check that these two subpaths completely lie in  $\GN \cup Y$. The weight of this path is  $|R'| = |sv \DIA \GN|  + |R[v,t]| = d(s,v,\GN)  + |R[v,t]|=|R[s,v] |\ +\ |R[v,t]|= |R| $.

\end{proof}


\subsection{One endpoint of $e$ is in $\GM$ and the other is in $\GN$ and $t \in \GG$}
\label{sec:middle}
This is an easy case as the $st$ path survives in $\GG \SM e$ as $st$ does not contain $e$. Thus, $st \DIA e = st$


%% file: fig_departing.tex
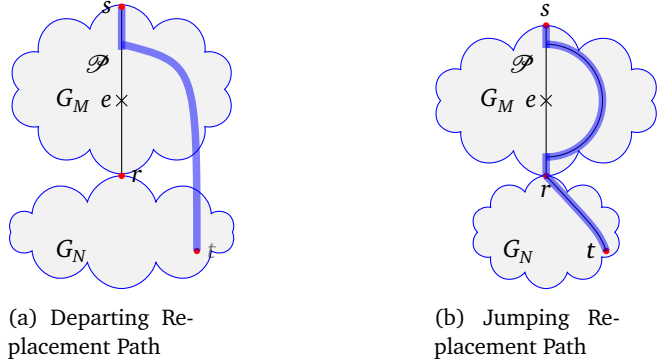
\begin{figure}[hpt!]
        \centering
        \begin{subfigure}[b]{0.15\textwidth}
                \centering
                \begin{tikzpicture}[scale=0.5]
                \node[cloud,
                draw =blue,
                text=cyan,
                fill = gray!10,
                minimum width = 6 cm,
                minimum height = 4.5cm,scale=0.5] (c) at (0,3.3) {};
                
                \node[cloud,
                draw =blue,
                text=cyan,
                fill = gray!10,
                minimum width = 6cm,
                minimum height = 3cm, scale=0.5] (c) at (0,-0.5) {};

                \draw (0,5.5) node[left] () {$s$} -- (0,1) node[right] () {$r$};
                \draw (0,5.5)  -- (0,4.35) [blue][line width=1mm,blue,opacity=0.5] ;
                \filldraw [red] (2,-1) circle (2pt);
                \filldraw [red] (0,1) circle (2pt);
                \filldraw [red] (0,5.5) circle (2pt);
                \draw (0,4.5) .. controls (2,4) .. (2,-1) node[right] {$t$} [blue][line width=1mm,blue,opacity=0.5];
                \draw (0,3) node[left] () {$e$} pic[rotate = 0] {cross=3pt};
                \draw (0,4) node[left] {$\PP$};
                \draw (-2,3) node[right] {$\GM$};
                \draw (-2,-1) node[right] {$\GN$};
                \end{tikzpicture}
                \caption{Departing Replacement Path}
        \end{subfigure}
        \hspace{3 cm}
        \begin{subfigure}[b]{0.15\textwidth}
                \begin{tikzpicture}[scale=0.5]
                
                \coordinate (s) at (0,13);
                \coordinate (r) at (0,9);
                \coordinate (t) at (1.6,7);
                \coordinate (u) at (0,12.5);
                \coordinate (v) at (0,9.6);
                \coordinate (e) at (0,6);
                
                \node[cloud,
                draw =blue,
                text=cyan,
                fill = gray!10,
                minimum width = 6 cm,
                minimum height = 4cm,scale=0.5] (c) at (0,11) {};

                \node[cloud,
                draw =blue,
                text=cyan,
                fill = gray!10,
                minimum width = 4cm,
                minimum height = 3cm,scale=0.5] (c) at (0,7.5) {};

                \draw (s) node[above] () {$s$} -- (r) node[below] () {$r$};
                
                \draw (-2,11) node[right] {$\GM$};
                \draw (-1.4,7) node[right] {$\GN$};
                
                \filldraw [red] (r) circle (2pt);
                \filldraw [red] (s) circle (2pt);
                \filldraw [red] (t) circle (2pt);
                \draw (0,11) node[left] () {$e$} pic[rotate = 0] {cross=3pt};
                \draw  (0,12.5) arc(90:-90:1.5);
                \draw  (0,12.5) arc(90:-90:1.5)[line width=1mm,blue,opacity=0.5];
                \draw (r) .. controls (1.4,7.5) .. (t) node[left] {$t$} ;
                \draw (0,12) node[left] {$\PP$};
                
                \draw  (s) -- ($(u) + (0,-0.1)$) [ line width=1mm,blue,opacity=0.5];
                \draw  (v) -- (r) [ line width=1mm,blue,opacity=0.5];
                \draw (r) .. controls (1.4,7.5) .. (t) node[left] {$t$}[line width=1mm,blue,opacity=0.5] ;

                \end{tikzpicture}
                \caption{Jumping Replacement Path}
        \end{subfigure}
        \caption{Departing and Jumping Replacement Path}
        \label{fig:departingjumping}
\end{figure}

%% file: fig_gmgn.tex
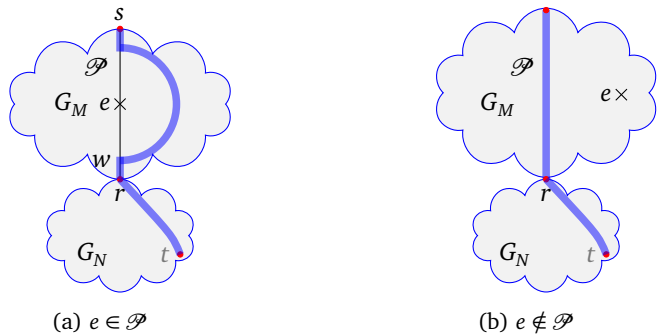
\begin{figure}[hpt!]
        \centering

        \begin{subfigure}[b]{0.15\textwidth}
                \begin{tikzpicture}[scale=0.5]
                
                \coordinate (s) at (0,13);
                \coordinate (r) at (0,9);
                \coordinate (t) at (1.6,7);
                \coordinate (u) at (0,12.5);
                \coordinate (v) at (0,9.6);
                \coordinate (e) at (0,6);
                
                \node[cloud,
                draw =blue,
                text=cyan,
                fill = gray!10,
                minimum width = 6 cm,
                minimum height = 4cm,scale=0.5] (c) at (0,11) {};

                \node[cloud,
                draw =blue,
                text=cyan,
                fill = gray!10,
                minimum width = 4cm,
                minimum height = 3cm,scale=0.5] (c) at (0,7.5) {};

                \draw (s) node[above] () {$s$} -- (r) node[below] () {$r$};
                
                \draw (0,9.5) node[left] () {$w$};
                
                \draw (-2,11) node[right] {$\GM$};
                \draw (-1.4,7) node[right] {$\GN$};
                
                \filldraw [red] (r) circle (2pt);
                \filldraw [red] (s) circle (2pt);
                \filldraw [red] (t) circle (2pt);
                \draw (0,11) node[left] () {$e$} pic[rotate = 0] {cross=3pt};
                \draw  (0,12.5) arc(90:-90:1.5)[line width=1mm,blue,opacity=0.5];
                \draw (0,12) node[left] {$\PP$};
                
                \draw  (s) -- ($(u) + (0,-0.1)$) [ line width=1mm,blue,opacity=0.5];
                \draw  (v) -- (r) [ line width=1mm,blue,opacity=0.5];
                \draw (r) .. controls (1.4,7.5) .. (t) node[left] {$t$}[line width=1mm,blue,opacity=0.5] ;

                \end{tikzpicture}
                \caption{ $ e \in \PP$}
        \end{subfigure}
                   \hspace{3 cm} \begin{subfigure}[b]{0.15\textwidth}
                \centering
                \begin{tikzpicture}[scale=0.5]
                \node[cloud,
                draw =blue,
                text=cyan,
                fill = gray!10,
                minimum width = 6 cm,
                minimum height = 4.5cm,scale=0.5] (c) at (0,3.3) {};
                
                \node[cloud,
                draw =blue,
                text=cyan,
                fill = gray!10,
                minimum width = 4cm,
                minimum height = 3cm,scale=0.5] (c) at (0,-0.5) {};

                \draw (0,1) node[below] () {$r$};
                \draw (0,5.5)  -- (0,1) [blue][line width=1mm,blue,opacity=0.5] ;
                \filldraw [red] (1.6,-1) circle (2pt);
                \filldraw [red] (0,1) circle (2pt);
                \filldraw [red] (0,5.5) circle (2pt);
                \draw (0,1) .. controls (1.4,-0.5) .. (1.6,-1) node[left] {$t$} [blue][line width=1mm,blue,opacity=0.5];
                \draw (2,3.25) node[left] () {$e$} pic[rotate = 0] {cross=3pt};
                \draw (0,4) node[left] {$\PP$};
                \draw (-2,3) node[right] {$\GM$};
                \draw (-1.5,-1) node[right] {$\GN$};
                \end{tikzpicture}
                \caption{$e \notin \PP $}
        \end{subfigure}
 
        \caption{ $e \in \GM $ and $ t \in \GN$}
         \label{fig:gmgn}
\end{figure}

%% file: fig_gmgmcase1.tex
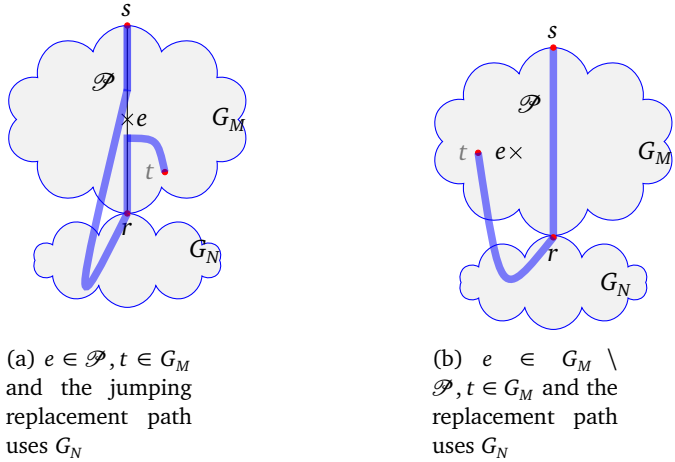
\begin{figure}[hpt!]

        \centering

        \begin{subfigure}[b]{0.15\textwidth}
                \begin{tikzpicture}[scale=0.5]
                
                \coordinate (s) at (0,13.5);
                \coordinate (r) at (0,8.5);
                \coordinate (t) at (1,9.6);
                \coordinate (u) at (0,12.5);
                \coordinate (v) at (0,9.6);
                \coordinate (e) at (0,6);
                
                \node[cloud,
                draw =blue,
                text=cyan,
                fill = gray!10,
                minimum width = 6.5 cm,
                minimum height = 5 cm,scale=0.5] (c) at (0,11) {};

                \node[cloud,
                draw =blue,
                text=cyan,
                fill = gray!10,
                minimum width = 5 cm,
                minimum height = 2.5 cm,scale=0.5] (c) at (0,7.25) {};

                \draw (s) node[above] () {$s$} -- (r) node[below] () {$r$};
                
                \draw (2,11) node[right] {$\GM$};
                \draw (1.4,7.5) node[right] {$\GN$};
                
                \filldraw [red] (r) circle (2pt);
                \filldraw [red] (s) circle (2pt);
                \filldraw [red] (t) circle (2pt);
                \draw (0,11) node[right] () {$e$} pic[rotate = 0] {cross=3pt};
                \draw (0,12) node[left] {$\PP$};
                
                \draw  (s) -- ($(u) + (0,-0.75)$) [ line width=1mm,blue,opacity=0.20];
                \draw (0,11.8) .. controls (-1.5,5.5) .. (r) [line width=1mm,blue,opacity=0.5] ;
                \draw  (s) -- ($(u) + (0,-0.75)$) [ line width=1mm,blue,opacity=0.5];
                
                \draw  (v)+(0,1) -- (r) [ line width=1mm,blue,opacity=0.5];
                \draw (0,10.5) .. controls (0.8,10.5) .. (t) node[left] {$t$}[line width=1mm,blue,opacity=0.5] ;

                \end{tikzpicture}
                \caption{ $ e \in \PP, t \in \GM$ and the jumping replacement path uses $\GN$}
        \end{subfigure}
                   \hspace{3 cm} \begin{subfigure}[b]{0.15\textwidth}
                \centering
                \begin{tikzpicture}[scale=0.5]
                \node[cloud,
                draw =blue,
                text=cyan,
                fill = gray!10,
                minimum width = 6.5 cm,
                minimum height = 5cm,scale=0.5] (c) at (0,4.3) {};
                
                \node[cloud,
                draw =blue,
                text=cyan,
                fill = gray!10,
                minimum width = 5cm,
                minimum height = 2.5cm,scale=0.5] (c) at (0,0.55) {};

                \draw (0,1.7) node[below] () {$r$};
                \draw (0,6.8)  -- (0,1.7) [blue][line width=1mm,blue,opacity=0.5] ;
                \filldraw [red] (-2,4) circle (2pt);
                \filldraw [red] (0,1.75) circle (2pt);
                \filldraw [red] (0,6.8) circle (2pt);
                \draw (0,1.75) .. controls (-1.4,0) .. (-2,4) node[left] {$t$} [blue][line width=1mm,blue,opacity=0.5];
                \draw (-1,4) node[left] () {$e$} pic[rotate = 0] {cross=3pt};
                \draw (0,5.3) node[left] {$\PP$};
                \draw (0,6.8) node[above] {$s$} ;
                \draw (2,4) node[right] {$\GM$};
                \draw (1,0.5) node[right] {$\GN$};
                \end{tikzpicture}
                \caption{$ e \in \GM \setminus \PP, t \in \GM$ and the replacement path uses $\GN$ }
        \end{subfigure}
 
        \caption{ $e \in \GM $ and $ t \in \GM$}
         \label{fig:gmgm}
\end{figure}

%% file: departing.tex
\section{Departing Replacement paths}
\label{sec:departing}

In the previous section, we have already found out a replacement path if it is jumping or if the edge failure is not on the primary path. In this section, we will try to find the best departing path after an edge failure. To this end, we define the following:
\begin{definition} (Candidate departing paths)
Let $e \in \PP$. $P$ is called the candidate departing path for $e$, if among all departing paths avoiding $e$, $P$ has the minimum length. In case there are many departing paths avoiding same edge with same length, we will break ties using Definition \ref{def:replacementpath}.
\end{definition}

Note that, a candidate departing path may or may not be a replacement path. In case it is, we call it a departing replacement path. Also, $P$ may be a candidate departing path for all edges in an interval, say $yz \in \PP$, but may be a replacement path for a sub-interval of $yz$.
With this definition in hand, we will now find all candidate departing paths.

\subsection{Finding all candidate departing paths}

In the previous section,  we added some weighted edges in the graph when we recurse. Thus, there might be two types of edges in the graph -- {\em weighted}  and {\em unweighted}. The weighted edges represent paths in the graph $\GG$, and the unweighted edges are present in $\GG$. $\GG$ contains only unweighted edges. However, the graph at an internal node of $\TTT$, say graph $\GAT$, may have weighted as well as unweighted edges.

In the ensuing discussion, let $\GAT$ be the graph processed by our algorithm at some internal node of $\TTT$. In the graph $\GAT$, we will assume that there is a source $s$. Let $\overline{\GG}$ be the parent of  $\GAT$.  In our algorithm, we partition $\overline{\GG}$ into two disjoint graphs and then recurse on it. If  $\GAT$ is the left child of $\overline{\GG}$, then it includes a set  $X$  of weighted edges added by us in  \Cref{subsec:MM}. Similarly, if $\GAT$ is the right child of $\overline{\GG}$ then it includes a set  $Y$  of weighted edges added by us in  \Cref{subsec:NN}. Using the separator lemma, we will find the vertex $r$ that  partitions $\BFS_{\GAT}(s)$ . Also, the primary path $\PP = sr$.

We now give an overview of our method to find candidate departing paths in $\GAT$. To this end, we will use the main idea in the paper of Gupta and Singh \cite{GuptaS18}.  In \cite{GuptaS18}, the authors proved \Cref{lem:guptaaditi}.
Though they did not mention it, the paths in the \Cref{lem:guptaaditi} look very much like the departing paths. Indeed, that lemma is more general than what the authors originally intended it to be. The authors show the above lemma for a specific vertex $t_s$, but a careful reading of the paper suggests that the above lemma is true for any vertex on the $st$ path. We now generalise this lemma.
However, there is one problem. The above result holds only for an unweighted graph, whereas $\GAT$ is weighted. Thus, we cannot prove the above lemma as it is. However, we will prove the following weaker lemma:

\begin{lemma}
\label{lem:guptaadapted}
Let $\GAT$ be the graph at  an internal node in the binary tree $\TTT$. Let $s$ be the source of $\GAT$. For a destination $t \in \GAT$, let $p$ be any vertex  on $st$ path. The number of replacement paths from $s$ to $t$ that avoid edges on $sp$ path  and also avoid  $p$ is $O(\sqrt n)$.
\end{lemma}

Some discussion is in order. In an ideal case, the number of replacement paths that avoid edges on $sp$ and also avoid $p$ should have been $\OO(\sqrt {n_{\GAT}})$. This result would have been similar to  \Cref{lem:guptaaditi}. Unfortunately, we cannot prove this result as $\GAT$ is weighted. However, if we just expand the weighted edge in the graph $\GAT$, then  we will get an unweighted graph. By expanding, we mean that for each weighted edge, add the path that represents that weighted edge. However, this process  may increase the number of vertices in the graph to $n$. Now, we can adapt   \Cref{lem:guptaaditi}. For an unweighted graph with $n$ vertices, this lemma guarantees that the number of replacement paths avoiding $p$ will be at most $\sqrt n$. Indeed, this is our result in  \Cref{lem:guptaadapted}. We prove this lemma in \Cref{adapted} as it is an extension of the proof in \cite{GuptaS18}.

 \Cref{lem:guptaadapted} implies that there are only $O(\sqrt n)$ replacement paths that have some special properties which are similar to the properties of departing paths.
So, our plan of action is as follows:

\begin{enumerate}
    \item Show that there are only $O(\sqrt n)$ candidate departing paths to $t$ in $\GAT$. This will be done by showing the similarity between the replacement paths in  \Cref{lem:guptaadapted} and candidate departing paths.

    \item Show that we can find the lengths of all the candidate departing paths in $\GAT$ in  $O(m_{\GAT} \sqrt n)$ time. Additionally, we show that we can store the lengths	 in a compact data structure of size $O(n_{\GAT}\sqrt n)$.  Given any query $\QU(s,t,e)$ to this data-structure such that $s,t \in \GAT$ and $e$ is on the primary path of $\GAT$, we can find the length of  corresponding candidate departing path in $\TL(1)$ time.
\end{enumerate}

This completes the overview of our algorithm for finding candidate departing paths.

\subsection{Similarity between candidate departing paths and replacement paths in  \Cref{lem:guptaadapted}}

 Let us first prove some simple results that will help us prove this section's main idea.
\begin{lemma}
\label{lem:departingproperty}
Let $t \in \GT$ and $p = \LCA(t,r)$. All the candidate departing paths from $s$ to $t$ avoiding edges on $sp$ path also avoid $p$. For any edge $e$ on $pr$ path , $st \DIA e = st$. 
\end{lemma}

\begin{proof}
Let $R = st  \DIA e$ where $e \in sp$ and $R$ is departing.  The detour of $R$ must start above $e$ on $\PP$.  Since $R$ is departing, it can not merge with the path $\PP$ again. So, it avoids $p$ also.

The $st$ path passes through $p$ and does not use any vertex of $\PP$ below $p$. So, removing any edge on $pr$ path  does not disturb $st$ path. So, for any edge $e \in pr$ , $st \DIA e = st$.

\end{proof}

We now show that the candidate departing paths and the replacement paths in  \Cref{lem:guptaadapted} have the same property.

\begin{lemma}
\label{lem:numberofdeparted}
For any vertex $t \in \GAT$, there are $\OO(\sqrt{n})$ candidate departing paths to $t$.
\end{lemma}
\begin{proof}
 \label{cdpnumbers}
 Remember that the candidate departing path avoids edges on path $\PP$ only. Let $p = \LCA(r,t)$. Using  \Cref{lem:departingproperty}, $st$ is the candidate departing path avoiding all edges in $pr$ path. Again, by  \Cref{lem:departingproperty}, all the candidate departing paths avoiding edges in $sp$ path  avoid $p$ too. But by  \Cref{lem:guptaadapted}, the number of such paths is $O(\sqrt n)$. This implies that there are $O(\sqrt n)$ candidate departing paths to $t$.

\end{proof}

Given the above lemma, we need to store $O(\sqrt n)$ candidate departing paths to $t$ in $\GAT$. Before designing an algorithm to find candidate departing paths, we first see how we plan to store these paths in a compact data structure.

\subsection{The data-structure at each node of $\TT$}
Let us first discuss a small technical detail that may be perceived as a problem but is not. Our graph $\GAT$ is weighted. It stands to reason that even the primary path $\PP$ in $\GAT$ maybe weighted. Since candidate departing paths are only for the faults on the primary path, it may not be clear what happens if the edge on the primary path is weighted. To this end, we show that on any $st$ path in $\GAT$, all edges except may be the first one, are unweighted.

\begin{lemma}
\label{lem:dsatT}
For $t \in \GAT$, except may be the first edge of $(st)_\GAT$, all other edges of $(st)_\GAT$ are unweighted.
\end{lemma}

\begin{proof}

   \label{structure}
  We will prove using induction on the nodes of $\TT$. In the root of $\TTT$, we have the graph $\GG$. Clearly, the graph $\GG$ satisfies the property of the lemma. Let us assume using induction that the parent of the graph $\GAT$ satisfies the property of the lemma. Let $\overline{\GG}$ be the parent of $\GAT$. Let $\overline{s}$ be the source in $\overline{\GG}$. Thus, in the graph $\overline{\GG}$, using the separator lemma, we find a $\overline{r}$ that divides $\overline{\GG}$ into two parts. The path from $\overline{s}$ to $\overline{r}$, say $\overline{\PP}$ is the primary path. By induction, only the first edge of the primary path may be weighted.  There are two cases:

\begin{enumerate}
\item $\GAT$ is the left child of $\overline{\GG}$.

In this case, the source of $\GAT$ is also $\overline{s}$.
For each $t \in \GAT$, it can be observed  that the path $(\overline{s}t)_{\overline{\GG}} = (\overline{s}t)_\GAT$. Using induction hypothesis, since $(\overline{s}t)_{\overline{\GG}}$ satisfies the statement of lemma, so does $(\overline{s}t)_\GAT$.
\item $\GAT$ is the right child of $\overline{\GG}$.

There are two cases. When $\overline{s} = \overline{r}$, then we fall back in point(1). So, let us look at the case when $\overline{s} \neq \overline{r}$.

For a $t \in \GAT$, by construction, $(\overline{s}t)_{\overline{\GG}}$ passes through $\overline{r}$. Thus, $(\overline{s}t)_{\overline{\GG}} = (\bar{s}\bar{r})_{\overline{\GG}} + (\overline{r}t)_{\overline{\GG}}$. In $\GAT$, we add a new source $s$. Also, we add a weighted edge  from $s$ to $\overline{r}$ in $\GAT$.  The weight of this edge is $|\bar{s}\bar{r} \DIA \GAT|_{\overline{\GG}}$. Also, $(\overline{r}t)_{\overline{\GG}} = (\overline{r}t)_\GAT$. This implies that there is path in $\GAT$ from $s$ to $t$, $(s,\overline{r})+ (\overline{r}t)_\GAT$. By induction, we claim that on this path except $(s,\overline{r})$, all the edges are unweighted.
\end{enumerate}

\end{proof}

Using the above lemma, all except the first edge of the primary path are unweighted. The weighted edges represent edges for which we have already found candidate departing paths at the parent or an ancestor of $\GAT$. Thus, we will only find candidate departing paths for unweighted edges in $\PP$.  In the ensuing discussion, whenever we mention a path avoiding an edge on the primary path, it will always refer to an unweighted edge.

 We now prove some simple lemmas that will help us build data structures for candidate departing paths.
\begin{lemma}
\label{lem:reldeparting1}
Let $R$ and $R'$ be two different candidate departing paths from $s$ to $t$ avoiding edges $e$ and $e'$  respectively on the path $\PP$. If $e$ lies above  $ e'$ on $\PP$, then $|R|>|R'|$.
\end{lemma}
\begin{proof}
The detour of the candidate departing $R$ starts before $e$, and once it departs, it does not merge with $\PP$ again.  As $e$ lies above  $e'$ on $\PP$, $R$  also avoids $e'$. If $|R| \leq |R'|$, then by \Cref{def:replacementpath}, $R$ must be  the  candidate departing path avoiding $e'$, leading to a contradiction. So, it must be the case that $|R|>|R'|$.
\end{proof}

\begin{lemma}
\label{lem:reldeparting2}
Let $R$ be a candidate departing path. Let $yz$ be the maximal subpath of $\PP$ such that $R$ is the candidate departing path for edges in $yz$. Then $\DP(R) = y$.	
\end{lemma}

\begin{proof}

\label{app:reldeparting2}
Since $R$ avoids the edges of $yz$, its detour starts at or above $y$ on the path $\PP$. Let us assume that $\DP(R)$ lies above $y$ on path $\PP$, say at $x$. This implies that $R$ also avoids all edges on $xy$ path. But $R$ is not the candidate departing path for edges in the path $xy$. Let $R'$ be the candidate departing path for an edge, say $e' \in xy$. Now, for an edge $e \in yz$, $R$ is the replacement path and $e'$ lies above $e$ on $\PP$. Using  \Cref{lem:reldeparting1}, $|R'| > |R|$. But in that case, $R$ should be the replacement path avoiding  $e'$ too. This leads to a contradiction. Thus, our assumption that $\DP(R)$ lies above $y$ is false and $\DP(R) = y$.
\end{proof}

\label{app:reldeparting3}

The above lemma states that if $R$ avoids edges in $yz$, then the detour of $R$ necessarily starts from $y$.
  We will now prove the contrapositive of the  \Cref{lem:reldeparting1}.

  \begin{lemma}
  \label{lem:reldeparting3}
Let $R$ and $R'$ be two different candidate departing paths from $s$ to $t$. If $|R|>|R'|$, then $\DP(R)$ lies above $\DP(R')$ on $\PP$.
\end{lemma}

\begin{proof}
Since $R'$ does not merge with $\PP$, $R'$ avoids all edges below $\DP(R')$ in $\PP$. Assume for contradiction, that $\DP(R')$ lies above $\DP(R)$.  This will imply that $R'$ also avoids all the edges avoided by $R$. But $|R| > |R'|$. Then, $R'$ should be the  candidate departing path for all the edges avoided by $R$ -- a contradiction. Thus, $\DP(R)$ lies above $\DP(R')$ on $\PP$.
\end{proof}

 We will now use the above lemmas and our deduction to build a compact data-structure for all candidate departing paths. To this end, we will store an array $\DEP(t)$ for each $t \in \GAT$. $\DEP(t)$ will store candidate departing paths from $s$ to $t$  in increasing order of their lengths. By  \Cref{lem:numberofdeparted}, there are  $\OO(\sqrt n)$ such paths. let us denote them by $R_1, R_2, \dots, R_k$ where $k = O(\sqrt n)$.  For any two consecutive candidate departing paths $R_i$ and $R_{i+1}$, using  \Cref{lem:reldeparting3} and \Cref{lem:reldeparting2}, we claim that $R_{i+1}$ is the candidate departing path avoiding edges in $[\DP(R_{i+1}),\DP(R_{i})]$ on the primary path $\PP$.
Since the size of $\DEP(t) = O(\sqrt n)$, the total size of $\DEP$ data-structure is bounded by $\OO(n_\GAT \sqrt n)$.

\begin{lemma}
$\sum_{t \in \GAT} \text{size of $\DEP(t)$} = O(n_\GAT\sqrt n)$
\end{lemma}
\subsection{Finding and storing all candidate departing paths efficiently}
Let us first describe the setting that will be used throughout this section. At an internal node $\GAT$ of $\TT$, we are planning to find all candidate departing paths from the source $s$. To this end, we will find a vertex $r$ that will divide $\SPT_{\GAT}(s)$ roughly equally. Also, $\PP =sr$.

To find all candidate departing path,  we will build an auxiliary graph $G$ which we will build incrementally. The source in this graph is   $(s)$. All other vertices in $G$ are tuples of the form  $(v,|R|)$, where $v \in \GAT \setminus \PP$ and $R$ is a candidate departing path to $v$.   During initialization, we will add $(v,|sv|_{\GAT})$ in $G$ for each $v \in \GAT \setminus \PP$. Also, there will be an edge from $(s)$ to $(v,|sv|_{\GAT})$ with weight $|sv|_{\GAT}$. We will show the following property at the end of our analysis.

\begin{prop}
\label{prop1}
 Let $R$ be the candidate departing path to $v$ avoiding edge on the subpath $yz$ of $\PP$.  Then, there is a vertex $(v,|R|)$ in
$G$. Moreover the shortest path from $(s)$ to $(v,|R|)$ in the graph $G$ is of length
$|R|$, that is  $|R|= |(s)(v,|R|)|_{G}$
\end{prop}

In  \Cref{lem:propertytrue}, we will show that \Cref{prop1} is true for all the nodes added during initialization. Also, we will create $\DEP(v)$ for each $v$ during intialization.
 We will store candidate departing paths in $\DEP(v)$ in increasing order of lengths. Given a departing path $R$, we will assume that we will store the following information about $R$ in $\DEP()$.
 \begin{enumerate}
\item The endpoints of $R$.
\item The weight of path $R$.
\item The last edge of $R$ and its weight.
\item $\DP(R)$.
\end{enumerate}

After initialization, we will run a variant of Dijkstra's algorithm in $G$.
To this end, we will construct a min-heap $\HH$ in which we will store all the departing paths that we have discovered at any point in the algorithm.
We use the first two points of \Cref{def:detour} to select the minimum  element from $\HH$,i.e., given two candidate departing paths $R$ and $R'$, $R$ is {\em smaller} than $R'$ if $|R| < |R'|$ or $|R| = |R'|$ and $\DEP(R)$ is closer to $s$ than $\DEP(R')$. If $\DEP(R) = \DEP(R')$, then we can break ties arbitrarily.

We now explain the adaptation of Dijkstra's algorithm. After initialization,for each $(v, |sv|)$, and for each neighbor neighbor $w$ of $v$, we add  the departing  path $sv + (v,w)$ in $\HH$ if $w \notin \PP$.
 Then, we go over all the elements of the heap till it is empty. Let us assume that $R$ is the minimum element of the heap and it ends at $v$ and $(u,v) \in \GAT$ is the last edge of $R$. This implies that $R$ was added in $\HH$ while processing a candidate departing path for $u$. Let this path be $R_u$. Thus, there is a node $(u,|R_u|)$ in $G$. We now need to decide whether $R$ is a candidate departing path for $v$.

To this end,  we look at the last candidate departing path added by us in $\DEP(v)$, let it be $Q$.  We then check if the $d_\GAT(s,\DP(R))$ is less than $d_\GAT(s,\DP(Q))$. If yes, then we have found a new candidate departing path to $v$ that avoids all edges in  $[\DP(R),\DP(Q)]$ of $\PP$. Thus, we will add the vertex $(v,|R|)$ in the graph with an edge of weight $wt_\GAT(u,v)$ from $(u,|R_u|)$. Also, for each neighbor $w$ of $v$, we will add the departing path $R+(v,w)$ to the heap if $w \notin \PP$. The pseudocode of the algorithm for finding all candidate departing path is given in algorithm 1.

\begin{algorithm}
\label{algo:cdpalgo}
\caption{Algorithm to construct $\DEP()$}

 create a vertex $(s)$\;

\For{$v \in \GAT \SM \PP$}
{

Let $R_v = |sv|_\GAT$\;   create a vertex $(v,|R_v|)$\;
   add edge between ($s)$ and $(v,|R_v|)$ with weight $|sv|_\GAT$\;
   append $R_v$ in  $\DEP(v)$\;
   \For{$(v,w) \in \GAT$ and $w \notin \PP$}
      {
      add the path $R_v+(v,w)$ in min-heap  $\HH$ \;
      }
}

\While{min-heap is non-empty}
{
  remove $R$ from the top of min-heap\;

Let us assume that $R$ ends at $v$ and the last edge of $R$ is $(u,v)$\;
Also assume that $R_u = R \SM (u,v)$ is a candidate departing path to $u$\;
  Let  $Q$ be the last candidate departing path added by us  in   $\DEP(v)$\;
  \If{$d_{\GAT}(s,\DP(R))<d_\GAT(s,\DP(Q))$}
  {
    create a vertex $(v,|R|)$\;
    add edge between $(u,|R_u|)$ and $(v,|R|)$ with weight $wt_\GAT(u,v)$ \;

    append the path $R$ in  $\DEP(v)$\;
    \For{$(v,w) \in \GAT$ and $w \notin \PP$}
    {
    add the departing path $R+(v,w)$  in min heap $\HH$\;
    }
 }

}

\end{algorithm}

\subsection{Correctness and running time of the algorithm storing candidate departing paths}
\label{subsection:correct}

We claim that, the time taken to construct the $\DEP()$ data-structure at a node of the binary tree with graph $\GAT$ is  $O(m_\GAT\sqrt n)$. Moreover, the size of the data-structure is $O(n_\GAT \sqrt n)$. In this section, we prove that  our algorithm stores correct lengths of all candidate departing paths.\begin{lemma}
	\label{lem:propertytrue}
Let $R$ be a candidate departing path to $v$ where $v \in \GAT \SM \PP$. Let $yz$ be the maximal subpath of $\PP$ such that $R$ is the candidate departing path for edges in $yz$. Then $(v,|R|) \in G$ and satisfies  \Cref{prop1}.
\end{lemma}

\begin{proof}
	We will prove the above lemma using induction on the weighted distance of a vertex from $(s)$ in $G$. During initialization, for each $v \in \GG \SM \PP$, we add a vertex $(v,|sv|_\GAT) \in G$. Also, the weight of the edge from $(s)$ to $(v,|sv|_\GAT)$ is $|sv|_\GAT$.   We claim that the statement of the lemma is true for the smallest candidate departing path to $v$. Indeed, using  \Cref{lem:departingproperty}, $(sv)_\GAT$ is the a replacement path for the edges in subpath $pr$ on $\PP$ where $p = \LCA_{\GAT}(v,r)$. Also,  $(sv)_\GAT$ is the smallest replacement path because it is the shortest path from $s$ to $v$ in $\GAT$.     Thus, the base case is true for all $v \in \GAT \SM \PP$.

Let us now assume that the statement of the lemma is true for all candidate departing paths to $v$ with length $< |R|$.  Let the second last vertex of $R$  be $u$.
Since $R$ is a candidate departing path, even $R \SM(u,v)$ is a candidate departing path. Since $R[s,u]$ has length less than $R$, by induction hypothesis, there is a vertex $(u,|R \SM(u,v)|)$ in $G$. Also there is at least one  replacement path in $\DEP(v)$  of weight less than $|R|$ -- as we have added $(sv)_\GAT$ in $\DEP(v)$. Let $Q$ be a candidate departing path of  largest weight less  than the weight of $R$. Let  us also assume that $Q$ avoids edge on subpath $zz' \in \PP$. Thus, using  \Cref{lem:reldeparting2}, $\DP(Q) = z$.  Since $|Q| < |R|$, using  \Cref{lem:reldeparting3}, $\DP(R)$ lies above $\DP(Q)$ on $\PP$.  Using the induction hypothesis,  there is a vertex $(v,|Q|)$ in $G$.

We will now show that our algorithm will add $(v,|R|)$ in $G$. There are four cases:

\begin{enumerate}
\item Our algorithm does not add any vertex for $v$ after $(v,|Q|)$

But our algorithm does process the vertex $(u,|R \SM(u,v)|)$. Thus, it will check each neighbour of $u$. When it checks the neighbor  $v$, it will  add the  departing path $R$ in the heap $\HH$. Thus, we will add the vertex $(v,|R|)$ in $G$ while processing $R$, leading to a contradiction.

\item Our algorithm adds a vertex $(v,|R'|)$ where $\DP(R)$ lies above   $\DP(R')$ on $\PP$

We claim that the weight of $R'$ cannot be less than the weight of $R$ as then  $R$ is not the candidate departing path for all the edges in $\DP(R')z$ subpath, contradicting the statement of the lemma. So let us assume that $|R| = |R'|$. But then $\DP(R)$ lies above $\DP(R')$ on $\PP$. Thus, the min-heap will give preference to the replacement path $R$ first, and our algorithm will make the vertex $(v,|R|)$. Again a contradiction.

\item Our algorithm adds a vertex $(v,|R'|)$ where $\DP(R)$ lies below  $\DP(R')$ on $\PP$

 Again, we claim that the weight of $R'$ cannot be less than the weight of $R$ as then it $R$ is not the replacement path for all edges in $yz$ subpath, contradicting the statement of the lemma. So let us assume that $|R| = |R'|$. But $\DP(R') $  is closer to $s$ than $\DP(R)$. Thus, $R'$ should be the  candidate departing path avoiding edges of $yz$. This contradicts our assumption that $R$ is the candidate departing path for all edges in $yz$.

\item Our algorithm adds a vertex $(v,|R'|)$ where $\DP(R) = \DP(R')$

$|R'|$ cannot be less than $|R|$ as otherwise our algorithm will give preference to path $R$. But, if  $|R| = |R'|$, then there is a vertex $(v, |R'|)= (v,|R|)$ in $G$. 

\end{enumerate}

So, we  add the  node $(v,|R|)$ in the graph $G$. At that moment, we also adds an edge from $(u,R\SM(u,v))$ to $(v,|R|)$ in $G$ with weight $wt_{\GAT}(u,v)$. Using induction, $|(s)(u,R\SM(u,v))|_G = |R\SM(u,v)|$. Thus, $|(s)(v,|R|)|_G = |R\SM(u,v)|\ +\ wt_\GAT(u,v) = |R|$. This completes our proof.

\end{proof}

 Let's determine the running time of our algorithm. Using \Cref{lem:numberofdeparted}, for each $v \in \GAT$, we make $O(\sqrt n)$ entries in $\DEP(v)$. In other words, we make $O(\sqrt n )$ nodes of type $(v ,\dot)$ in $G$.   Whenever we add a node $(v,|R|)$  in $G$, we see all the edges of $v$.   This implies that the total time taken to process all the vertices of $v$ in $G$ is $O( \sqrt n \deg_\GAT(v))$. Summing it over all the vertices gives us the bound of $O(m_\GAT\sqrt n)$.  Using a similar calculation, the total size of our data-structure for  $\GAT$ is $O(n_\GAT \sqrt n)$. Thus, we claim the following lemma:

\begin{lemma}
\label{lem:depconstruction}
The time taken to construct the $\DEP()$ data-structure at a node of the binary tree with graph $\GAT$ is  $O(m_\GAT\sqrt n)$. Moreover, the size of the data-structure is $O(n_\GAT \sqrt n)$.
\end{lemma}

\subsection{Querying for a candidate departing path}
\label{query}
In this section, we describe how to find a candidate departing path using our data-structure $\DEP()$. Let $t \in \GAT \SM \PP$  and $e \in \PP$ be an edge on $st$ path. Let  $st \DIA e$ be   a candidate departing path, then we can find it using  algorithm given in \Cref{app:queryCDP}.

In this algorithm, we perform a binary search in $\DEP(t)$ to find two consecutive paths $R$ and $Q$ such that $e$ lies in the interval [$\DEP(R), \DEP(Q)]$  of $\PP$. Using  \Cref{lem:reldeparting3} and \Cref{lem:reldeparting2},  $R$ is the candidate departing path avoiding $e$.
The pseudocode of the query algorithm to find the Candidate Departing path here.
\label{app:queryCDP}
\begin{algorithm}
\caption{$\QD(s,t,e)$}


 $R \xleftarrow{}$ binary search in the array of $\DEP(t)$ to find a two consecutive path $Q$ and $R$ such that $e$ lies in interval  $[\DP(R),\DP(Q)]$ of $\PP$ \;
return{ $|R|$}
\label{algo:queryd_1}
\end{algorithm}

\begin{algorithm}[H]
\caption{$\QU(s,t,e,\TT)$}

 $mindist \xleftarrow{} \infty $\;

  \tcc{ \Cref{sec:gmgn} }

\If{ $e \in \GM,t \in \GN $}
{
    \If{ $e \in \PP$}
    {
     \tcc{if $st \DIA e$ happens to be departing}
     $mindist \leftarrow \QD(s,t,e)$ using $\DEP()$ data-structure at $\TT$\;
     \tcc{if $st \DIA e$ happens to be jumping}

     $mindist \xleftarrow{} \min\{mindist, |sr \diamond e|+|rt|\}$
    }
    \Else{
    $mindist \leftarrow  |st|$
    }
}
\tcc{ \Cref{subsec:MM}}

\If{ $e \in \GM,t \in \GM $}
{
\tcc{ \Cref{sec:gmgmp}}
\If{ $e \in \PP $}
{
      \tcc{if $st \DIA e$ happens to be departing}
      $mindist \leftarrow \QD(s,t,e)$ using $\DEP()$ data-structure at $\TT$\;
      \tcc{if $st \DIA e$ happens to be jumping}
      $mindist \xleftarrow{} \min (mindist, |sr \DIA e|+|rt| ) $

   }
    \tcc{ \Cref{sec:gmgnnp}}
    \If{ $e \in \GM \setminus \PP $}
    {
      $mindist \xleftarrow{}  \QU(s,t,e,\text{left child of $\TT$})$
   }

}

\tcc{ \Cref{sec:gngm}}
\If{ $e \in \GN, t \in \GM $}
{
    $mindist \leftarrow  |st|$\;
}

\tcc{ \Cref{subsec:NN}}
\If{$ e \in \GN,t \in \GN,  $}
{
   $mindist \leftarrow \QU(s,t,e,\text{right child of $\TT$}$)\;

}
\tcc{ \Cref{sec:middle}}
\If{one endpoint of $e$ is in $\GM$ and other in $\GN$}
{
$mindist \leftarrow  |st|$\;
}

return $mindist$\;
\label{alg:query}
\end{algorithm}

 \section{Construction time, size and query time of the SDO(1) } 
 In this section, we show that the construction time of our algorithm is $\TL(m \sqrt n)$. We also bound the size of the data-structure of our algorithm by $\TL(n\sqrt n)$. We also design a query algorithm with a query time $\TL(1)$. This proves the main result of the paper. 
 
 At the root of $\TT$,  except for the recursions,  we claim that constructing all other data-structures  take $O(m \sqrt{n})$ time. This is beacuse, the construction time is dominated by the time to construct $\DEP()$, which using  \Cref{lem:depconstruction}, is $O(m\sqrt n)$.    At the second level of the tree $\TT$, we have two nodes. In the left child, we have the graph $\GM \cup X$. This graph  has $m_{\GM} + n_{\GM}$ edges and $n_{\GM}$ vertices. Again applying  \Cref{lem:depconstruction}, the time taken to construct all the data-structures in the left child of root is $(m_{\GM}+n_{\GM})\sqrt n$. In the right child of the root, we have the graph $\GN \cup Y$. This graph has $m_{\GN} + n_{\GN}$ edges and $n_{\GN}+1$ vertices. The $+1$ is for the new root in $\GN$. Again applying  \Cref{lem:depconstruction}, the time taken to construct all the data-structures in the right child of root is $(m_{\GN}+n_{\GN})\sqrt n$. Thus, the total time taken at the second level of $\TT$ is    = $(m_{\GM}+n_{\GM}) \sqrt{n} + (m_{\GN}+n_{\GN}) \sqrt{n}$. Since $m_{\GM}+m_{\GN} \le m$ and $n_{\GM} + n_{\GN} = n+1$, the total time taken is  $\le (m+n+1) \sqrt n$. Note that $n_{\GM} + n_{\GN} = n+1$ because $r$ is shared both by  $\GM$\ and $\GN$.  Since, the number of nodes in $\TT$ is $O(n)$,  we claim that the number of vertices shared by sibling graphs at any level of $\TT$  is $O(n)$. Similar to the second level, we claim that the time taken at  level $\ell$ is $\TL((m +n +$ \#nodes shared at level $\ell$$)\sqrt  n) = \TL((m+n)\sqrt n)$.  We can assume that our graph $\GG$ is connected as  we need not even process a component that is not reachable from our source $s$. Thus, the previous running time bound is equal to $\TL(m \sqrt n)$. Since the height of the tree is $\TL(1)$, the total time taken to construct our data-structure is $\TL(m\sqrt n)$.
  Using the same argument, we can bound the size of the data-structure of our algorithm by $\TL(n\sqrt n)$.

  \subsection{The Query Algorithm}
  
  In this section, we will design our query algorithm that will take $s,t,e$ as its parameter. Additionally, it also takes the root of the tree $\TT$ as a parameter which contains data structures of the main graph $\GG$.  The algorithm then basically goes over all the possible cases described in   \Cref{sec:detail} (Please see Algorithm 3). Also, the algorithm compares that output with the best candidate departing path given by Algorithm 2 and returns the minimum among them.

  The reader can see that the time taken by the  algorithm (excluding recursion) is $\TL(1)$. Since, at each step in this algorithm, we either go to the left child of a node in the tree $\TT$ or to the right child, the number of recursive steps in this algorithm is $\TL(1)$. This implies that the running time of the algorithm is $\TL(1)$. Thus, we have proven the main theorem of the paper.

%% file: appendix.tex
\appendix





\section{Proof of Lemma \ref{lem:guptaadapted}}
\label{adapted}

In lemma \ref{lem:guptaadapted}, we have a graph $\GAT$ at some internal node of $\TT$. The source of $\GAT$ is $s$, and the primary path is $\PP =sr$. A weighted edge in $\GAT$ represents a path in the parent of $\GAT$, say $\overline{\GG}$. First of all, we recursively expand all weighted edges in $\GAT$. By expanding, we mean that we will replace a weighted edge  $e \in \GAT$ with its corresponding path in $\overline{\GG}$.  Similarly, a weighted edge $e'$ in $\overline{\GG}$ represents a path in its parent. So, we will again recursively replace $e'$ with a path. The reader can see that this process will lead us to a graph in which each edge $e$ of $\GAT$ is replaced by a path that contains only unweighted edges. However, this transformed graph may have $n$ vertices. Henceforth, we will assume that we are working with this transformed unweighted version of $\GAT$.

Even in the transformed version of $\GAT$, the shortest $st$ path remains the same -- since we have just expanded weighted edges. Since we now have an unweighted graph, we will follow the proof of \cite{GuptaS18}. We first define some terms:
\begin{itemize}

    \item Let $\RR$ be the set of all replacement paths that avoid edges on $sp$ path and avoid $p$ too. At the end of the proof, we will show that $|\RR| = O(\sqrt n)$.
    \item If $R \in \RR$ is a replacement path from $s$ to $t$ avoiding some edge, then we define the sets $(<R)$ and $(>R)$ as the set of all  replacement paths from $s$ to $t$ avoiding edges on the path $st$ which has length less than the length of $R$ and greater than the length of $R$ respectively.
    \item  \Det $(R)$ is the sub-path of $R$ from the vertex it leaves $st$ to the vertex where it merges back to $st$ path.
Thus, $\Det(R) = R \SM st$.
    \item \UNIQUE $(R)$ is the prefix of $R$ which does not intersect with any detours in $\bigcup_{R' \in (>R)}$ \Det $(R')$.
\end{itemize}

We now state the following lemma from \cite{GuptaS18} that will be useful later:

\begin{lemma}(Lemma 9 of \cite{GuptaS18})
\label{lem:lessthanR}
 If $R\in \RR$ is a replacement path from $s$ to $t$ avoiding $e$ such that $|R|=|st|+l$ where $l \geq 0$, then the size of the set $(<R)$ is $\leq l$.

\end{lemma}

Now, let us go through the $st$ replacement paths in $\RR$ in decreasing order of lengths. We can observe that the number of replacement paths with   $\UNIQUE \geq \sqrt{n}$ will be  $\le \sqrt{n}$ as the vertices on the detours must be mutually disjoint in \UNIQUE. Now, let $R$ be the first such replacement path such that \UNIQUE$(R) < \sqrt{n}$. We show that number of replacement paths in the set $(<R)$ will be $\OO(\sqrt{n})$.  Now, we have the following cases.

\begin{itemize}
    \item \textbf{ \Det$(R)$ does not intersect with detour of any path in $(>R)$}

Let the detour of $R$ starts at $a$ and ends at $b$. Thus $\UNIQUE(R) =ab$ and $|ab|  \le \sqrt n$. Using Lemma \ref{lem:reldeparting3}, the detour of every  replacement path avoiding edges between $a$ and $p$ from the set $(<R)$ must go through $a$ -- since the detour of these paths start below $a$. Now, $ |R \setminus sa |=|ab|+|bt| \leq |ab|+|at| < |at|+\sqrt{n}  $. So using lemma \ref{lem:lessthanR}, the number of paths in the set $ \{R' \setminus sa | R' \in (<R) \}$ is $\le \sqrt{n}$. So, number of replacement paths with length less than $R$ will be $\leq \sqrt{n}$.

    \item \textbf{ \Det$(R)$ intersects with detour of a path in $(>R)$}
    \input{fig_oracle}

     Let, the first path in $(>R)$ which intersects $R$ be $R'$. Let the detour of $R$ starts at $a$ and ends at $b$. Let the detour of $R'$ starts at $a'$ and ends at $b'$. Let $R$ and $R'$ intersect at point $c$. This implies that $\UNIQUE(R)=ac$ and $|ac| \le \sqrt n$. (See Figure \ref{fig:singlesecondcase})

\begin{align*}
|sa'|+|a'c|+|cb'|+|b't| &\leq  |sa'|+|a'c|+|ca|+|at| \\
|cb'|+|b't| &\leq  |ca| +|at| \\
|ac|+|cb'|+|b't| &\leq 2|ca|+|at|
\end{align*}

Now, the left hand side of the above inequality represents the path  $R \setminus sa$. So $|R \SM sa| = |ac|+|cb'|+|b't| $. Since $|ca| \leq \sqrt{n}$, $|R \setminus sa| \leq 2 \sqrt{n}+ |at|$. Again, using Lemma  \ref{lem:lessthanR}, the cardinality of the set of replacement paths $ \{R' \setminus sa | R' \in (<R) \}$ will be $\le 2 \sqrt{n}$.

\end{itemize}
Thus, we have shown that the number of replacement paths in $\RR$ with  $\UNIQUE \ge \sqrt n$ is $O(\sqrt n)$. Moreover, once we find a path in $\RR$ with $\UNIQUE < \sqrt n$, then there are $O(\sqrt n)$ path in $R$ left to be processed. Thus, the number of replacement paths in $\RR$ is $O(\sqrt n)$.
 This completes the proof of the lemma.

%% file: fig_oracle.tex
 \begin{figure}[hpt!]
        \centering
        \begin{subfigure}[b]{0.3\textwidth}
                \centering
                \begin{tikzpicture}[scale=2.5]
                
                \definecolor{dgreen}{rgb}{0.0, 0.5, 0.0}
                \begin{scope}[xshift=0cm]
                \coordinate (s) at (0,2);
                \coordinate (t) at (0,0);
                \coordinate (ts) at (0,0.4);
                \coordinate (b1) at (0,.2);
                
                \coordinate (a1) at (0,1.5);
                \coordinate (a) at (0,1);
                \coordinate (v) at (0,0.7);
                \coordinate (v1) at (0,1.2);
                \coordinate (c) at (-0.585,0.7);
                
                \draw[thick](s)--(t);
                \node[above] at (s){$s$};
                \node[below] at (t){$t$};
                \node[right] at (a1){$a$};
                \node[right] at (b1){$b$};

                \draw[blue,thick] (a1) to[out=180,in=180,distance=.8cm] node[pos=0.3,left]
                {\scriptsize  $R$}(b1);
                
                \node at (v1){$\times$};
                \node[right] at (v1){$e$};
                
                \node at (ts){$\times$};
                \node[right] at (ts){$p$};
                \end{scope}
                
                \end{tikzpicture}
                
                \caption{ $\Det(R)$ does not intersect with detour of any path in $(>R)$ }
                \label{fig:singlesecondcase}
        \end{subfigure}
        \hspace{3 cm}
        \begin{subfigure}[b]{0.3\textwidth}
                \begin{tikzpicture}[scale=2.5]
                
                \definecolor{dgreen}{rgb}{0.0, 0.5, 0.0}
                \begin{scope}[xshift=0cm]
                \coordinate (s) at (0,2);
                \coordinate (t) at (0,0);
                \coordinate (ts) at (0,0.4);
                \coordinate (b1) at (0,.2);
                
                \coordinate (a1) at (0,1.5);
                \coordinate (a) at (0,1);
                \coordinate (v) at (0,0.7);
                \coordinate (v1) at (0,1.2);
                \coordinate (c) at (-0.585,0.7);
                
                \draw[thick](s)--(t);
                \node[above] at (s){$s$};
                \node[below] at (t){$t$};
                \node[right] at (a1){$a'$};
                \node[right] at (a){$a$};
                \node[right] at (b1){$b'$};
                \node[left] at (c){$c$};

                \draw[blue,thick] (a1) to[out=180,in=180,distance=.8cm] node[pos=0.3,left]
                {\scriptsize  $R'$}  (b1);
                
                \draw[red,thick] (a) to[out=180,in=80]
                node[pos=0.4,above]
                {\scriptsize  $R$}  (c);
                
                \node at (v1){$\times$};
                \node[right] at (v1){$e'$};
                
                \node at (v){$\times$};
                \node[right] at (v){$e$};
                
                \node at (ts){$\times$};
                \node[right] at (ts){$p$};
                \end{scope}
                
                \end{tikzpicture}
                
                \caption{ $\Det(R)$  intersects with detour of a path in $(>R)$ }
                \label{fig:singlesecondcase}
        \end{subfigure}
        \caption{Two cases according to \Det$(R)$ intersects or not with detour of a path in $(>R)$}

\end{figure}
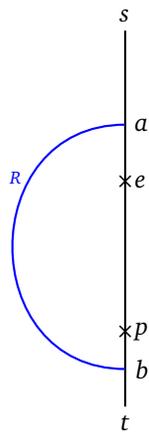
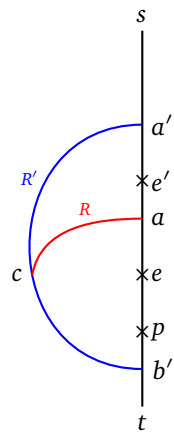

%% file: main.bbl
\begin{thebibliography}{10}

\bibitem{BenderF00}
Michael~A. Bender and Martin Farach{-}Colton.
\newblock The {LCA} problem revisited.
\newblock In {\em {LATIN} 2000: Theoretical Informatics, 4th Latin American
  Symposium, Punta del Este, Uruguay, April 10-14, 2000, Proceedings}, pages
  88--94, 2000.

\bibitem{Bernstein2009}
Aaron Bernstein and David Karger.
\newblock A nearly optimal oracle for avoiding failed vertices and edges.
\newblock In {\em Proceedings of the forty-first annual ACM symposium on Theory
  of computing}, pages 101--110. ACM, 2009.

\bibitem{Bernstein2008}
Aaron Bernstein and David~R. Karger.
\newblock Improved distance sensitivity oracles via random sampling.
\newblock In Shang{-}Hua Teng, editor, {\em Proceedings of the Nineteenth
  Annual {ACM-SIAM} Symposium on Discrete Algorithms, {SODA} 2008, San
  Francisco, California, USA, January 20-22, 2008}, pages 34--43. {SIAM}, 2008.

\bibitem{BiloCFS21}
Davide Bil{\`{o}}, Sarel Cohen, Tobias Friedrich, and Martin Schirneck.
\newblock Near-optimal deterministic single-source distance sensitivity
  oracles.
\newblock In {\em Accepted in Annual European Symposium on Algorithms, ESA
  2021}, 2021.

\bibitem{Bilo2016}
Davide Bil{\`{o}}, Luciano Gual{\`{a}}, Stefano Leucci, and Guido Proietti.
\newblock Multiple-edge-fault-tolerant approximate shortest-path trees.
\newblock In {\em 33rd Symposium on Theoretical Aspects of Computer Science,
  {STACS} 2016, February 17-20, 2016, Orl{\'{e}}ans, France}, pages
  18:1--18:14, 2016.

\bibitem{BodwinGPW17}
Greg Bodwin, Fabrizio Grandoni, Merav Parter, and Virginia~Vassilevska
  Williams.
\newblock Preserving distances in very faulty graphs.
\newblock In {\em 44th International Colloquium on Automata, Languages, and
  Programming, {ICALP} 2017, July 10-14, 2017, Warsaw, Poland}, pages
  73:1--73:14, 2017.

\bibitem{ChechikC19}
Shiri Chechik and Sarel Cohen.
\newblock Near optimal algorithms for the single source replacement paths
  problem.
\newblock In {\em Proceedings of the Thirtieth Annual {ACM-SIAM} Symposium on
  Discrete Algorithms, {SODA} 2019, San Diego, California, USA, January 6-9,
  2019}, pages 2090--2109, 2019.

\bibitem{ChechikCFK17}
Shiri Chechik, Sarel Cohen, Amos Fiat, and Haim Kaplan.
\newblock {(1} + {$\epsilon$})-approximate \emph{f}-sensitive distance oracles.
\newblock In {\em Proceedings of the Twenty-Eighth Annual {ACM-SIAM} Symposium
  on Discrete Algorithms, {SODA} 2017, Barcelona, Spain, Hotel Porta Fira,
  January 16-19}, pages 1479--1496, 2017.

\bibitem{Chechik2010}
Shiri Chechik, Michael Langberg, David Peleg, and Liam Roditty.
\newblock Fault tolerant spanners for general graphs.
\newblock {\em {SIAM} J. Comput.}, 39(7):3403--3423, 2010.

\bibitem{ChechikM20}
Shiri Chechik and Ofer Magen.
\newblock Near optimal algorithm for the directed single source replacement
  paths problem.
\newblock In Artur Czumaj, Anuj Dawar, and Emanuela Merelli, editors, {\em 47th
  International Colloquium on Automata, Languages, and Programming, {ICALP}
  2020, July 8-11, 2020, Saarbr{\"{u}}cken, Germany (Virtual Conference)},
  volume 168 of {\em LIPIcs}, pages 81:1--81:17. Schloss Dagstuhl -
  Leibniz-Zentrum f{\"{u}}r Informatik, 2020.

\bibitem{Demetrescu2008}
Camil Demetrescu, Mikkel Thorup, Rezaul~Alam Chowdhury, and Vijaya
  Ramachandran.
\newblock Oracles for distances avoiding a failed node or link.
\newblock {\em {SIAM} J. Comput.}, 37(5):1299--1318, 2008.

\bibitem{10.5555/1496770.1496826}
Ran Duan and Seth Pettie.
\newblock Dual-failure distance and connectivity oracles.
\newblock In {\em Proceedings of the Twentieth Annual ACM-SIAM Symposium on
  Discrete Algorithms}, SODA '09, page 506–515, USA, 2009. Society for
  Industrial and Applied Mathematics.

\bibitem{DuanP10}
Ran Duan and Seth Pettie.
\newblock Approximating maximum weight matching in near-linear time.
\newblock In {\em Proceedings of the 2010 IEEE 51st Annual Symposium on
  Foundations of Computer Science}, FOCS '10, pages 673--682, Washington, DC,
  USA, 2010. IEEE Computer Society.

\bibitem{GrandoniW12}
Fabrizio Grandoni and Virginia~Vassilevska Williams.
\newblock Improved distance sensitivity oracles via fast single-source
  replacement paths.
\newblock In {\em Foundations of Computer Science (FOCS), 2012 IEEE 53rd Annual
  Symposium on}, pages 748--757. IEEE, 2012.

\bibitem{GuptaJM20}
Manoj Gupta, Rahul Jain, and Nitiksha Modi.
\newblock Multiple source replacement path problem.
\newblock In Yuval Emek and Christian Cachin, editors, {\em {PODC} '20: {ACM}
  Symposium on Principles of Distributed Computing, Virtual Event, Italy,
  August 3-7, 2020}, pages 339--348. {ACM}, 2020.

\bibitem{GuptaK17}
Manoj Gupta and Shahbaz Khan.
\newblock Multiple source dual fault tolerant {BFS} trees.
\newblock In {\em 44th International Colloquium on Automata, Languages, and
  Programming, {ICALP} 2017, July 10-14, 2017, Warsaw, Poland}, pages
  127:1--127:15, 2017.

\bibitem{GuptaS18}
Manoj Gupta and Aditi Singh.
\newblock Generic single edge fault tolerant exact distance oracle.
\newblock In {\em 45th International Colloquium on Automata, Languages, and
  Programming, {ICALP} 2018, July 9-13, 2018, Prague, Czech Republic}, pages
  72:1--72:15, 2018.

\bibitem{Hershberger2001}
John Hershberger and Subhash Suri.
\newblock Vickrey prices and shortest paths: What is an edge worth?
\newblock In {\em 42nd Annual Symposium on Foundations of Computer Science,
  {FOCS} 2001, 14-17 October 2001, Las Vegas, Nevada, {USA}}, pages 252--259,
  2001.

\bibitem{MalikMG89}
Kavindra Malik, Ashok~K Mittal, and Santosh~K Gupta.
\newblock The k most vital arcs in the shortest path problem.
\newblock {\em Operations Research Letters}, 8(4):223--227, 1989.

\bibitem{NardelliPW03}
Enrico Nardelli, Guido Proietti, and Peter Widmayer.
\newblock Finding the most vital node of a shortest path.
\newblock {\em Theor. Comput. Sci.}, 296(1):167--177, 2003.

\bibitem{NisanR01}
Noam Nisan and Amir Ronen.
\newblock Algorithmic mechanism design.
\newblock {\em Games and Economic Behavior}, 35(1-2):166--196, 2001.

\bibitem{Parter2015}
Merav Parter.
\newblock Dual failure resilient {BFS} structure.
\newblock In {\em Proceedings of the 2015 {ACM} Symposium on Principles of
  Distributed Computing, {PODC} 2015, Donostia-San Sebasti{\'{a}}n, Spain, July
  21 - 23, 2015}, pages 481--490, 2015.

\bibitem{Parter2013}
Merav Parter and David Peleg.
\newblock Sparse fault-tolerant {BFS} trees.
\newblock In {\em Algorithms - {ESA} 2013 - 21st Annual European Symposium,
  Sophia Antipolis, France, September 2-4, 2013. Proceedings}, pages 779--790,
  2013.

\bibitem{10.1145/2438645.2438646}
Oren Weimann and Raphael Yuster.
\newblock Replacement paths and distance sensitivity oracles via fast matrix
  multiplication.
\newblock {\em ACM Trans. Algorithms}, 9(2), March 2013.

\end{thebibliography}
